\documentclass[preprint,12pt]{elsarticle}

\usepackage{amsmath, amsthm, amscd, amsfonts, amssymb, graphicx,graphics, color,subfig,natbib}
\usepackage[bookmarksnumbered, colorlinks, plainpages]{hyperref}
\usepackage[ruled,vlined,linesnumbered]{algorithm2e}

\usepackage[usenames,dvipsnames]{pstricks} 
\usepackage{epsfig} 
\usepackage{pst-grad} 
\usepackage{pst-plot} 
\usepackage[space]{grffile} 
\usepackage{etoolbox} 
\makeatletter 
\patchcmd\Gread@eps{\@inputcheck#1 }{\@inputcheck"#1"\relax}{}{} 
\makeatother

\usepackage{lineno}
\usepackage{setspace}

\newtheorem{theorem}{Theorem}[section]

\newtheorem{proposition}[theorem]{Proposition}

\theoremstyle{definition}

\newtheorem{example}[theorem]{Example}

\theoremstyle{remark}
\newtheorem{remark}[theorem]{Remark}

\numberwithin{equation}{section}

\journal{Journal Name}
\begin{document}

\begin{frontmatter}


\title{ Marshall-Olkin exponential shock model covering all range of dependence}



\author{H. A. Mohtashami-Borzadaran$^1$, M. Amini$^1$, H. Jabbari$^1$, A. Dolati$^2$}

\address{$^1$Department of Statistics, Ferdowsi University of Mashhad, Mashhad, Iran\\ 
$^2$Department of Statistics, Faculty of Mathematics, Yazd University, Yazd, Iran
}

\begin{abstract}
In this paper, we present a new Marshall-Olkin exponential shock model. The new construction method gives the proposed model further ability to allocate the common joint shock on each of the components, making it suitable for application in fields like reliability and credit risk. The given model has a singular part and supports both positive and negative dependence structure. Main dependence properties of the model is given and an analysis of stress-strength is presented. After a performance analysis on the estimator of parameters, a real data is studied. Finally, we give the multivariate version of the proposed model and its main properties.  
\end{abstract}

\begin{keyword}
Shock model \sep Marshall-Olkin model \sep Negative and positive dependence\\
\MSC[2010] 60E05, 60E15, 62N05.


\end{keyword}

\end{frontmatter}


\section{Introduction}
The univariate exponential distribution is known for its applications in different fields such as reliability, telecommunication,
hydrology, medical sciences and environmental science (see \citet{balakrishnan2018exponential}). One of the main multivariate extensions of exponential distribution was given by \citet{marshall1967multivariate}. They constructed their model based on shock models. Let $T_1 \sim E(\theta_1)$, $T_2 \sim E(\theta_2)$, $T_{12} \sim E(\theta_{12})$, then the Marshall-Olkin (MO) shock model is achieved by
\begin{equation}\label{MO.construction}
(X,Y)=\big( \min \{T_1,T_{12} \}, \min \{T_2,T_{12} \} \big).
\end{equation}
They showed that their extension verifies the lack of memory property in the multivariate case. The MO exponential distribution has both singular and continuous parts in its density and covers positive dependence structure. The MO model has a broad application in reliability (see \citet{cherubini2015marshall}), finance and actuarial science (see \citet[Chapter 7]{elouerkhaoui2017credit} and \citet{lindskog2003common}). For instance, \citet{lindskog2003common} applied the MO model to credit risk. They stated the MO model allows the applicability of dependence among  the shock arrival times while it preserves exponentially distributed observed lifetimes. Also, \citet{cherubini2017gumbel} expressed that within the concept of credit risk and financial crisis, the MO model gives an important and flexible tool to study and represent systemic crises.  First, they asserted that, the MO model uses the unobserved shocks which are subject to each individual or common among individuals. Second, the common shocks are used for simultaneous default of the elements in the cluster stated as individual subsets disclosed to the same common factor. Third, the MO model preserves the marginal exponential distribution used for observed default times.   
\par Many bivariate and multivariate extensions of exponential distribution have been introduced and applied in reliability. \citet{esary1974multivariate} characterized a multivariate exponential distribution and derived a positive dependence condition for multivariate distributions with exponential minimum. \citet{raftery1984continuous} proposed a continuous multivariate exponential distribution which can model a full range of correlation and attain the Fr{\'e}chet bounds in the bivariate case. \citet{tawn1990modelling} introduced a multivariate exponential distribution that arises from limiting joint distribution of normalized component-wise maxima or minima. \citet{lin1993multivariant} used a shared-load model of the multivariate exponential distribution to describe the characteristics of dependent redundancies. The multivariate exponential distributions with constant failure rates has been given by \citet{basu1997multivariate}. The multivariate power exponential distribution was given by \citet{gomez1998multivariate}. 
\citet{cui2007analytical} propose an analytical method for reliability of coherent systems with dependent components based on MO model.
\citet{fan2009multivariate} proposed a multivariate exponential survival tree procedure that splits data using the score test statistic derived from a parametric exponential frailty model, which allows for fast evaluation of splits. \citet{kundu2009bivariate} estimated the parameter of a new bivariate exponential distribution using EM-algorithm. 
\citet{li2011generalized} gave a generalization of bivariate MO distributions, which the well-known MO model includes as special case. \citet{kundu2013bayes} did a Bayesian estimation for the MO bivariate Weibull distribution. \citet{bayramoglu2014reliability} used a MO model system by taking into account the system structure. \citet{kundu2014multivariate}  introduced a multivariate proportional reversed hazard model obtained from the MO copula. \citet{cha2017multivariate} derived a multivariate exponential distribution model based on dependent dynamic shock models. \citet{al2018weighted} generated a multivariate weighted exponential distribution to analyze failure time data. 
Recently, \citet{mohtashami-borzadaran_jabbari_amini_2020} extended the MO shock model with a distortion function that made the previous MO model more applicable.
\par Consider the construction in (\ref{MO.construction}). From a shock model point of view the MO model has limitation in terms of shock equality in the common shock $T_{12}$, that is the shock $T_{12}$ is likely to be equal on the two components $X_1$ and $X_2$. Our new model solves this issue by giving the new model concrete ability to set the random percentage amount of common shock on each of components $X_1$ and $X_2$. From a distribution theory point of view, most of the bivariate and multivariate extensions of exponential distribution have positive dependence structure and rarely have negative dependence structure. Recently,  \citet{mohsin2014new} proposed a new bivariate exponential distribution for modeling moderately negative dependent data. Here, we propose a new multivariate exponential shock model that can have negative dependence structure too.  
In Section 2, we give the new shock model and explain it's flexibility comparing to MO model. In Section 3, main properties of the proposed model such as dependence structure, association measures, tail dependence measures and stress-strength index are given. Section 4 focuses on the estimation of parameters for the new model which is challenging since it has a singular part. After that a performance analysis of the estimators are analyzed. Section 5 focuses on the application of real data given in \citet{mohsin2014new} and we show that the new model is much better model. Finally, we present the new multivariate MO shock model and obtain its main properties in Section 6. 
\section{The proposed model}
Given three independent exponential random variables $T_1$, $T_2$ and $T_{12}$ and an arbitrary standard uniform random variable $U$ that is independent of $T_1$, $T_2$ and $T_{12}$. Suppose $T_1 \sim E(\theta_1)$, $T_2 \sim E(\theta_2)$, $T_{12} \sim E(\theta_{12})$. Let $\alpha_{12}$ (taking values $\pm 1$) indicate the dependence structure of the model where  $\alpha_{12}=+1$  concludes positive and  $\alpha_{12}=-1$ gives negative dependence structure. When $\alpha_{12}=+1$, set $T_{12}^* (\alpha_{12})=T_{21}^* (\alpha_{12})=F_{T_{12}}(U)$ or $T_{12}^* (\alpha_{12})=T_{21}^* (\alpha_{12})=F_{T_{12}}(1-U)$. If $\alpha_{12}=-1$, put  $T_{12}^* (\alpha_{12})= F_{T_{12}}(U),T_{21}^* (\alpha_{12})=F_{T_{12}}(1-U)$ or $T_{12}^* (\alpha_{12})= F_{T_{12}}(1-U),T_{21}^* (\alpha_{12})=F_{T_{12}}(U)$ where $F_{T_{12}}$ is the corresponding distribution function of $T_{12}$. Then, the bivariate MO random vector $( R,S)$ covering all degree of dependence is
\begin{equation}\label{eeer}
(R,S)=\Big( \min \{T_1,T_{12}^* (\alpha_{12})\} , \min \{T_2,T_{21}^* (\alpha_{12}) \} \Big). 
\end{equation}
Clearly, when $\alpha_{12}=+1$, the vector $(R,S)$ reduces to 
$$(R,S)=\Big( \min \{T_1,T_{12} \},\min \{T_2,T_{12} \} \Big), $$
which is the well-known MO model given in \citet{marshall1967multivariate} that has positive dependence structure. When $\alpha_{12}=-1$,  the random vector $(R,S)$ gives a new MO model with negative dependence structure (see Proposition \ref{prop11}) which is called the bivariate negative MO model, denoted by $BNMO(\theta_1,\theta_2,\theta_{12})$. This model is obtained by
\begin{equation}\label{vec.model}
(R,S)=(\min \{T_1,F_{T_{12}}^{-1} (U)\},\min \{T_2,F_{T_{12}}^{-1} (1-U) \}),
\end{equation}
or
\begin{equation}\label{vec.model1}
(R,S)=(\min \{T_1,F_{T_{12}}^{-1} (1-U)\},\min \{T_2,F_{T_{12}}^{-1} (U) \}).
\end{equation}
Throughout this paper, we focus on the $(R,S)$ given in (\ref{vec.model}). 
\begin{figure}[h!]
\centering
\includegraphics[width=0.66\textwidth , height=0.28\textheight]{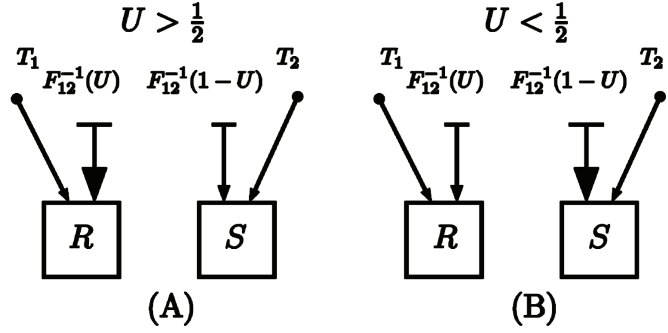}
\caption{Shock models based on the construction in (\ref{vec.model}).  {\bf (A)} When $U>0.5$, R is more likely to receive shock than S and {\bf (B)} when $U<0.5$, S is more likely to receive shock than R.}
\label{shockmodel.fig}
\end{figure}
A similar construction for a bivariate Poisson model has been given by \cite{genest2018new}. 
\par The interpretation for this construction is different to the well-known MO model. Consider Figure \ref{shockmodel.fig} given based on the relation (\ref{vec.model}). If $U>\frac{1}{2}$ then the dependent shock is more likely to be powerful on the first component $R$ and If $U<\frac{1}{2}$ the dependent shock is more likely to be powerful on the second component ($S$). \\
The survival function of both vectors (\ref{vec.model}) and (\ref{vec.model1}) for $\theta_1,\theta_2,\theta_{12}>0$ is 
\begin{eqnarray}
\bar{F}_{R,S}(r,s)&=&P( T_1>r,T_2>s,F_{T_{12}}(r)<U<1-F_{T_{12}}(s)), \nonumber \\
 &=& e^{-\theta_1 r} e^{-\theta_2 s} \big(  e^{-\theta_{12} r}+e^{-\theta_{12} s}-1 \big), \label{surv.model}
\end{eqnarray}
where $e^{-\theta_{12} r}+e^{-\theta_{12} s} \geq 1$.
This model has a singular part at $e^{-\theta_{12} r}+e^{-\theta_{12} s}=1$. The probability density function of (\ref{vec.model}) when $e^{-\theta_{12} r}+e^{-\theta_{12} s}>1$ is
\begin{equation}\label{dens.model}
f_{R,S}(r,s)= e^{-\theta_1 r -\theta_2 s} \Big( \theta_2 (\theta_1+\theta_{12}) e^{-\theta_{12}r}+\theta_1(\theta_2+\theta_{12})e^{-\theta_{12}s}-\theta_1 \theta_2 \Big),
\end{equation}
and based on \citet[ Theorem 1.1, pp. 15]{joebook1997}, for $e^{-\theta_{12} r}+e^{-\theta_{12} s}=1$ we have 
\begin{eqnarray}
h(r)&=& \lim_{s \rightarrow g^+ (r)} P(S \leq s  | R=r) -\lim_{s \rightarrow g^- (r)} P(S \leq s  | R=r) \nonumber \\
&=&\lim_{s \rightarrow g^- (r)} \frac{ \frac{\partial		}{\partial r} \bar{F}_{R,S}(r,s) }{ \frac{\partial		}{\partial r} \bar{F}_{R}(r) }  -\lim_{s \rightarrow g^+ (r)} \frac{ \frac{\partial		}{\partial r} \bar{F}_{R,S}(r,s) }{ \frac{\partial		}{\partial r} \bar{F}_{R}(r) }  \nonumber \\
&=& \frac{\theta_{12}	}{\theta_{1} + \theta_{12}} \big( 1- \exp \{ - \theta_{12} r \} \big)^{\theta_{2}/\theta_{2}},\nonumber 
\end{eqnarray}
where $g(r)= \frac{-1}{\theta_{12}} \ln (1- \exp \{ - \theta_{12} r \})$.\\
 The following statement gives the probability of the singular part.
\begin{proposition}
Set $\alpha:=\frac{\theta_{12}}{\theta_1 +\theta_{12}}$, $\beta:=\frac{\theta_{12}}{\theta_2 +\theta_{12}}$ and let $(R,S) \sim BNMO(\theta_1,\theta_2,\theta_{12})$. Then 
\[P(e^{-\theta_{12} R}+e^{-\theta_{12} S}=1)=Beta(\frac{1}{\alpha},\frac{1}{\beta}),\]
where $Beta(a,b)=\int_0^1 x^{a-1} (1-x)^{b-1} dx$.
\end{proposition}
\begin{proof}
Based on the construction in (\ref{vec.model}), we get 
\[ [e^{-\theta_{12} R}+e^{-\theta_{12} S}=1] \equiv [F_{T_{12}}(R)>U,F_{T_{12}}(S)>1-U]. \]
So, by conditioning w.r.t. $U=u$, we have
\begin{eqnarray}
P(e^{-\theta_{12} R}+e^{-\theta_{12} S}=1) &=& \int_0^1 P(F_{T_{12}}(R)>u) P(F_{T_{12}}(S)>1-u) du, \nonumber \\
&=& \int_0^1 u^{\frac{\theta_1}{\theta_{12}}} (1-u)^{\frac{\theta_2}{\theta_{12}}} du, \nonumber \\
&=& Beta(\frac{1}{\alpha},\frac{1}{\beta}). \nonumber
\end{eqnarray}
\end{proof}
\begin{remark}
If $\theta_1 = \theta_2 = \theta_{12}$ or equivalently $\alpha=\beta=\frac{1}{2}$,  then 
\[P(e^{-\theta_{12} R}+e^{-\theta_{12} S}=1)=\frac{1}{6}.\]
\end{remark}
\par The random vectors $(X_1,Y_1)$ and $(X_2,Y_2)$ can be compared in terms of their dependence structure via the upper orthant (UO) order. For any two vectors such as $(X_1,Y_1),(X_2,Y_2)$, we say $(X_1,Y_1)$ is less than $(X_2,Y_2)$ in UO order and write $(X_1,Y_1) \prec_{UO} (X_2,Y_2) $ whenever $\bar{F}_{X_1,Y_1} (x,y) \leq \bar{F}_{X_2,Y_2} (x,y)$ for all $x,y$. 
\begin{proposition}
Let $(R,S) \sim BNMO(\theta_1,\theta_2,\theta_{12})$ and $(R^{'},S^{'}) \sim BNMO(\theta_1,\theta_2,\theta_{12}^{'})$. If $\theta_{12} \leq \theta_{12}^{'}$ then $(R^{'},S^{'}) \prec_{UO} (R,S)$.
\end{proposition}
\begin{proof}
For any $r,s,\theta_1,\theta_2>0$ and $\theta_{12} \leq \theta_{12}^{'}$, we have 
\begin{eqnarray}
\bar{F}_{R,S} (r,s)&=&  e^{-\theta_1 r} e^{-\theta_2 s} (e^{-\theta_{12} r}+e^{-\theta_{12} s}-1) \nonumber \\
&\geq &  e^{-\theta_1 r} e^{-\theta_2 s} (e^{-\theta_{12}^{'} r}+e^{-\theta_{12}^{'} s}-1) \nonumber \\
&\geq & \bar{F}_{R^{'},S^{'}} (r,s). \nonumber
\end{eqnarray}
Hence, $(R^{'},S^{'}) \prec_{UO} (R,S)$ and this completes the proof.
\end{proof}
\section{Some properties}
In this section, we present some properties of BNMO model such as dependence structure, association measures, tail dependence measures and stress-strength index.
\subsection{Dependence structure}
Let $(X,Y)$ be a random vector with survival function $\bar{F}$. The pair $(X,Y)$ is said to be right corner set decreasing, denoted by $RCSD(X,Y)$, whenever for any $x_1<x_2$ and $y_1<y_2$ we have 
\[\bar{F}(x_1,y_1)\bar{F}(x_2,y_2)-\bar{F}(x_1,y_2)\bar{F}(x_2,y_1) \leq 0,\]
that is equivalent to 
\[\frac{\partial^2}{\partial r \partial s} \ln(\bar{F} (r,s)) \leq 0.\] 
$RCSD(X,Y)$ implies negative dependence structures like $RTD(X|Y)$, $RTD(Y|X)$ and $NQD(X,Y)$ (for more information see \citet{nelsen2007introduction}). The following statement specifies the dependence structure of the proposed model.
\begin{proposition}\label{prop11}
If $(R,S) \sim BNMO(\theta_1,\theta_2,\theta_{12})$, then we have $RCSD(R,S)$.\\
\end{proposition}
\begin{proof}
For all $\theta_1,\theta_2,\theta_{12} \in R$, we obtain
\begin{equation*}\label{dep.struc1}
\frac{\partial^2}{\partial r \partial s} \ln(\bar{F}_{R,S} (r,s)) =
\frac{-\theta_{12}^2e^{-\theta_{12}r-\theta_{12}s}}{(1-e^{-\theta_{12}r}-e^{-\theta_{12}s})^2} \leq 0,
\end{equation*}
that implies $RCSD(R,S)$ and the proof is complete.
\end{proof}

\subsection{Association measures and tail dependence}
For every pair $(R,S)$ with survival function $\bar{F}$, some famous measures of association are Kendall's tau $\tau=4E\big(\bar{F}(X,Y)\big)-1$ and Spearman's rho $\rho_s=12 \int_{(0,\infty)^2} (\bar{F} (x,y)-\bar{F}_1 (x)\bar{F}_2 (y))f_1 (x) f_2(y)dxdy$. Also, the lower and upper tail dependence coefficient $\lambda_L$  and $\lambda_U$ are defined by $\lambda_L= \lim_{t \rightarrow 0^+} P[X \leq F_1^{-1}(t)|Y \leq F_2^{-1}(t)]$ and $\lambda_U= \lim_{t \rightarrow 1^-} P[X>F_1^{-1}(t)|Y>F_2^{-1}(t)]$, respectively (see \citet{nelsen2007introduction}). Association measures $\tau$ and $\rho_s$ did not have closed form, so we plotted their variation for different values of  $\alpha=\frac{\theta_{12}}{\theta_1 + \theta_{12}}$ and $\beta=\frac{\theta_{12}}{\theta_2 + \theta_{12}}$. Based on Figure \ref{tau.fig}, as the value of $\alpha,\beta \rightarrow 1$ the value of dependence measure $\tau$ decreases to -1 and the dependency becomes stronger. Also, Figure \ref{spear.fig} illustrates that strength of dependence increases to $\rho_s=-1$ as $\alpha,\beta \rightarrow 1$. Figure \ref{rho to tau} shows that the value of $\rho_s/\tau \rightarrow 1.5$ as $\alpha,\beta \rightarrow 0$. This shows that as the dependency decreases to independence the value of $\rho_s/\tau \rightarrow 1.5$ and becomes lower if the dependency increases despite the fact we can't prove this theoretically. For the tail dependence, we prove the following statement.
\begin{proposition}
If $(R,S) \sim BNMO(\theta_1,\theta_2,\theta_{12})$, then $\lambda_L=\lambda_U=0$.
\end{proposition}
\begin{proof}
Let $\alpha=\frac{\theta_{12}}{\theta_1 + \theta_{12}}$ and $\beta=\frac{\theta_{12}}{\theta_2 + \theta_{12}}$. For every $\alpha,\beta \in (0,1)$,  we have
\[ \bar{F}_{R,S}(F_R^{-1}(t),F_S^{-1}(t))=(1-t)^{2-\alpha-\beta} \Big( (1-t)^\alpha + (1-t)^\beta -1 \Big). \]
So, 
\begin{eqnarray}
\lambda_L &=&  \lim_{t \rightarrow 0^+} P[X \leq F_1^{-1}(t)|Y \leq F_2^{-1}(t)] \nonumber\\
                     &=&  \lim_{t \rightarrow 0^+} \frac{1}{t} \big( 2t-1 +  \bar{F}_{R,S}(F_R^{-1}(t),F_S^{-1}(t)) \big)  \nonumber\\
                     &=&  \lim_{y \rightarrow 1^-}  \frac{1}{1-y} \big( 1-2y + y^{2-\alpha-\beta}(y^\alpha + y^\beta -1) \big)=0. \nonumber
\end{eqnarray} 
Also,
\begin{eqnarray}
\lambda_U &=& \lim_{t \rightarrow 1^-} P[X>F_1^{-1}(t)|Y>F_2^{-1}(t)] \nonumber\\
&=&  \lim_{t \rightarrow 1^-} \frac{1}{1-t} \big(\bar{F}_{R,S}(F_R^{-1}(t),F_S^{-1}(t) \big)  \nonumber\\
&=&  \lim_{y \rightarrow 0^+}  \frac{1}{y} \big(  y^{2-\alpha-\beta}(y^\alpha + y^\beta -1) \big)=0. \nonumber
\end{eqnarray} 
So the proof is complete.
\end{proof}

\begin{figure}
\centering
\subfloat{{\includegraphics[width=0.5\textwidth]{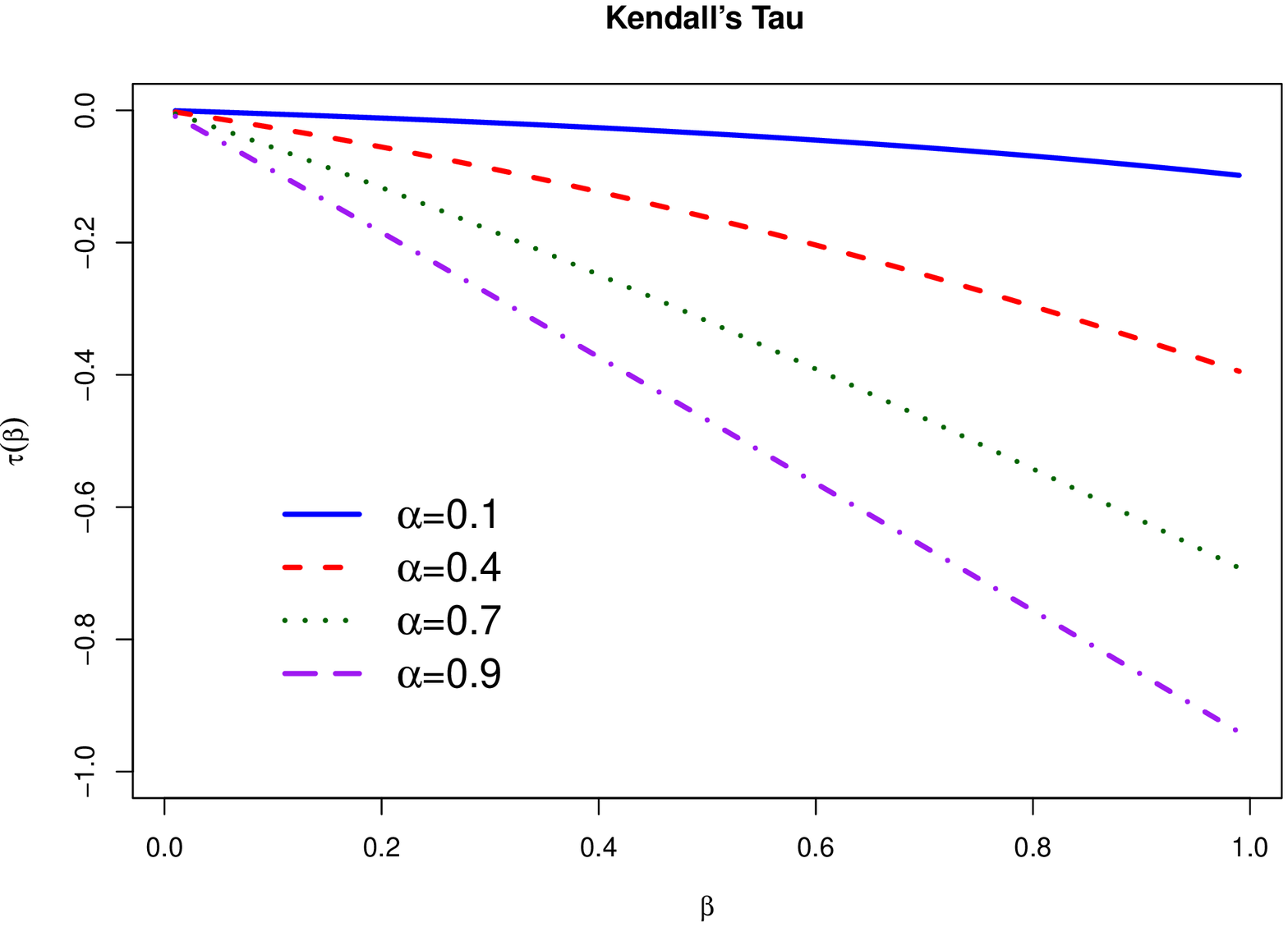} }}
\subfloat{{\includegraphics[width=0.5\textwidth]{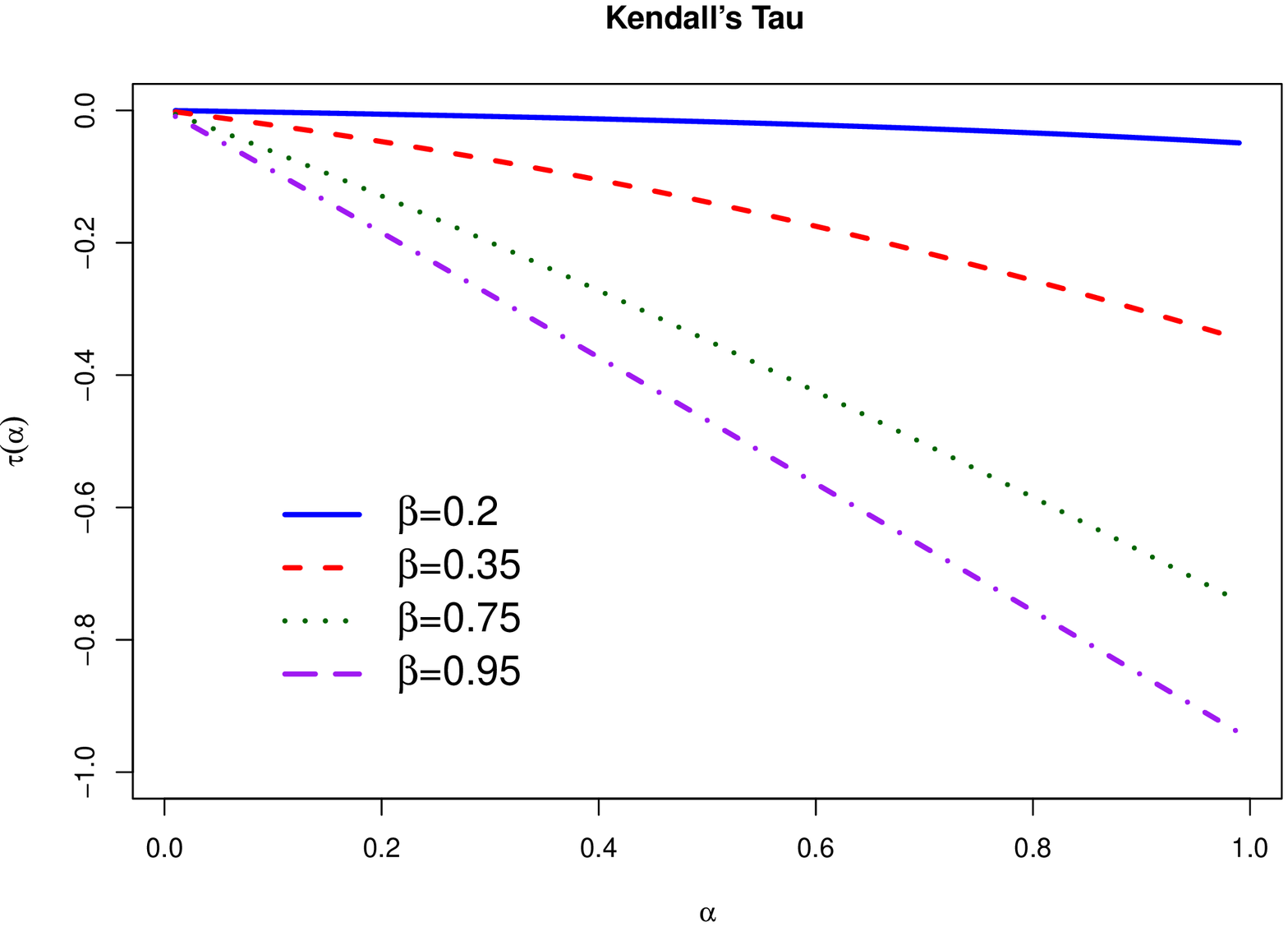} }}
\caption{Kendall's tau against different dependence parameters,	 $\alpha$ and $\beta$.}
\label{tau.fig}
\end{figure}

\begin{figure}
\centering
\subfloat{{\includegraphics[width=0.5\textwidth]{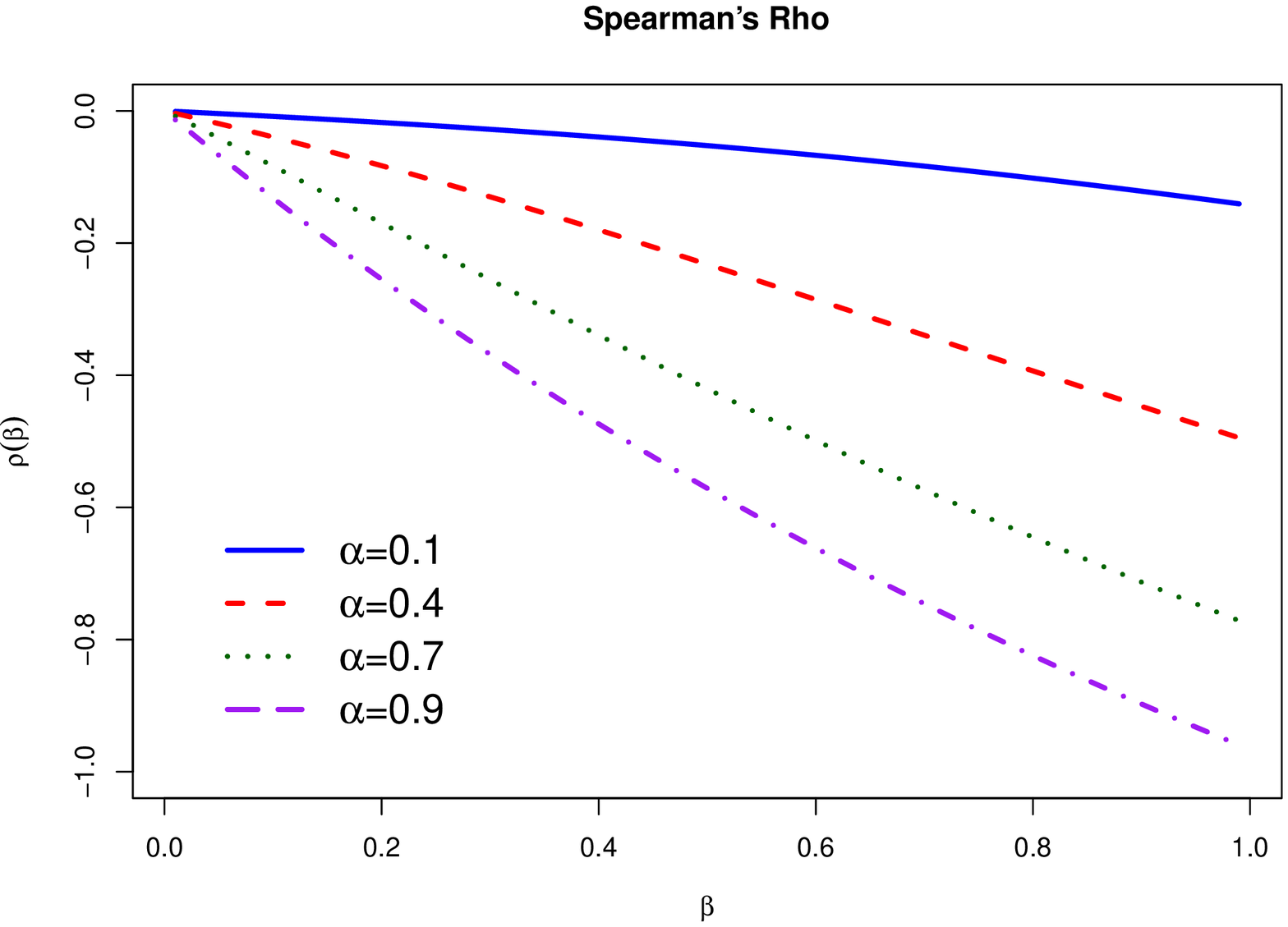} }}
\subfloat{{\includegraphics[width=0.5\textwidth]{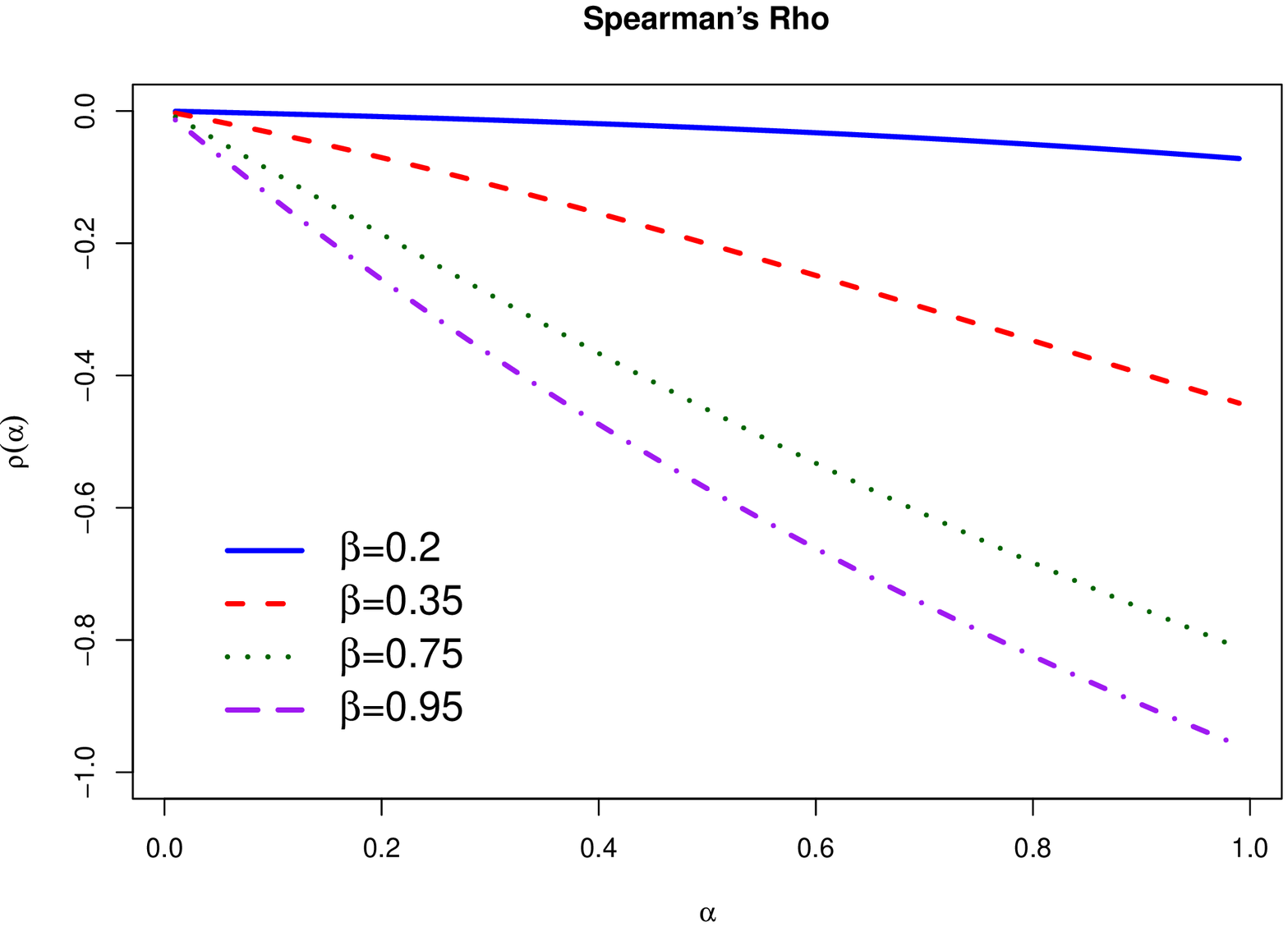} }}
\caption{Spearman's rho against different dependence parameters, $\alpha$ and $\beta$.}
\label{spear.fig}
\end{figure}

\begin{figure}[h!]
\centering
\includegraphics[width=0.5\textwidth]{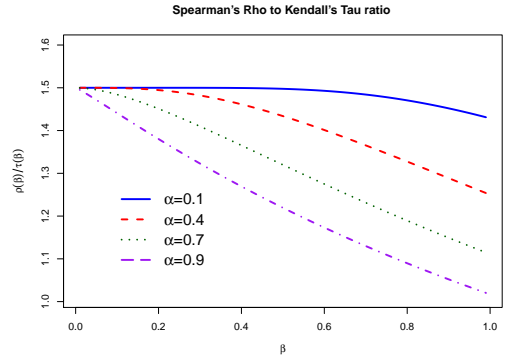}
\caption{Spearman's rho to Kendall's tau ratio for different values of $\alpha$ and $\beta$.}
\label{rho to tau}
\end{figure}

\subsection{Stress-strength index}
In the context of reliability stress-strength model can be described as an analysis of reliability for a system in terms of random variables $X$ representing stress (supply) experienced by the system and $Y$ representing the strength (demand) of the system available to tolerate the stress. The system fails when the stress exceeds the strength. So, $R=P(X<Y)$ is the reliability considering the failure mode described by the stress-strength relation. the stress-strength index can be computed in terms of competing risk given in \citet{shih2016bivariate}. Under the competeing risk models, failure times $X$ and $Y$ are called latent failure times. Based on the failure time $T=\min(X,Y)$ and failure cause $C=1$ if $X<Y$ or $C=2$ if $X>Y$, we define the sub distribution functions as 
\[F^*(1,t)=P(C=1,T \leq t)=\int_0^t f^* (1,z)dz,\]
and
\[F^*(2,t)=P(C=2,T \leq t)=\int_0^t f^* (2,z)dz,\]
where $f^* (1,t)=- \partial \bar{F}(x,y)/ \partial x |_{x=y=t}$ and $f^* (2,t)=- \partial \bar{F}(x,y)/ \partial y |_{x=y=t}$ that are called sub-density functions. Then, the stress-strength index is given by $P(X<Y)=F^* (1,\infty)$ and $P(X>Y)=F^* (2,\infty)$. According to the competing risk model, the stress-strength index for the proposed model is obtained in the following statement.
\begin{proposition}
Let $(R,S) \sim BNMO(\theta_1,\theta_2,\theta_{12})$, then 
\[P(R<S)=\frac{2\theta_1 + \theta_{12}}{\theta_1 + \theta_2 +\theta_{12}}-\frac{\theta_1}{\theta_1 + \theta_2},\] 
or equivalently
\[P(R<S)=\frac{2\beta -\alpha \beta}{\beta + \alpha -\alpha \beta} - \frac{\beta - \alpha \beta}{\beta + \alpha -2\alpha \beta}.\]
\end{proposition}

\begin{proof}
Based on (\ref{surv.model}), we have
\[f^* (1,t) = - \frac{\partial \bar{F}_{R,S} (r,s)}{\partial r}|_{r=s=t} =\theta_1 e^{-(\theta_1 + \theta_2)t}(2e^{-\theta_{12}t}-1)+
\theta_{12}e^{-(\theta_1 + \theta_2+\theta_{12})t}.\]
So,
\begin{eqnarray}
F^* (1,t)&=&\int_0^t f^* (1,u)du \nonumber \\
&=& \frac{2\theta_1 + \theta_{12}}{\theta_1 + \theta_2 +\theta_{12}} (1-e^{-(\theta_1 + \theta_2+\theta_{12})t}) - \frac{\theta_1}{\theta_1 + \theta_2} (1-e^{-(\theta_1 + \theta_2)t}). \nonumber
\end{eqnarray}
Thus,
\begin{eqnarray}
P(R<S)&=& \lim_{t \rightarrow +\infty} F^* (1,t) \nonumber \\
&=& \frac{2\theta_1 + \theta_{12}}{\theta_1 + \theta_2 +\theta_{12}}-\frac{\theta_1}{\theta_1 + \theta_2}. \nonumber 
\end{eqnarray}
Using $\alpha=\frac{\theta_{12}}{\theta_1 +\theta_{12}}$ and $\beta=\frac{\theta_{12}}{\theta_2 +\theta_{12}}$ we get the second statement.
\end{proof}

\begin{remark}
If $\theta_1=\theta_2$ or equivalently $\alpha=\beta$, then $P(R<S)=\frac{1}{2}$.
\end{remark}
\begin{figure}
\centering
\includegraphics[width=0.75\textwidth]{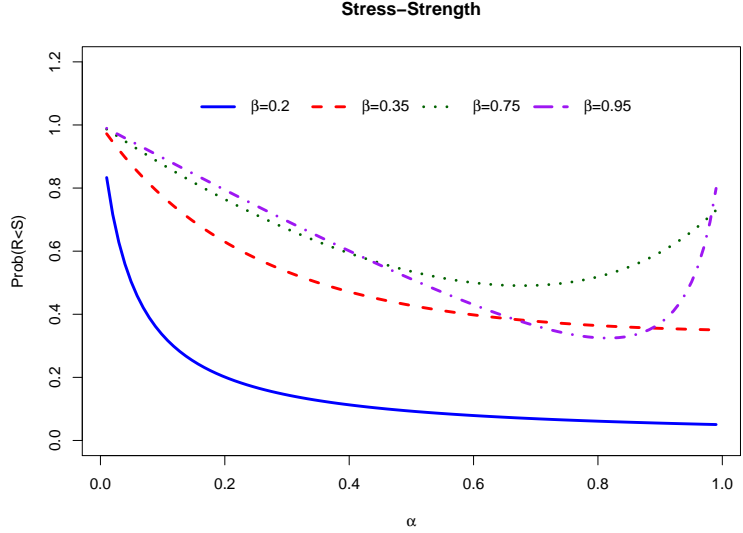}
\caption{Stress-strength index for different values of dependence parameters, $\alpha$ and $\beta$.}
\label{stress.fig}
\end{figure}
Figure \ref{stress.fig} illustrates the stress-strength index for different values of $\alpha$ and $\beta$.
\section{Estimation and simulation}

\subsection{Random number generation}
Simulating random numbers is essential in understanding the behaviour of a model. In order to generate random numbers from $BNMO(\theta_1 , \theta_2 , \theta_{12})$, the following algorithm is given.

\begin{algorithm}[H]
\label{alg1}
\caption{Random number generation from $BNMO(\theta_1 , \theta_2 , \theta_{12})$}
\begin{itemize}
\item[Step 1.] Generate three independent random variables $T_i \sim E(\theta_i)$ for $i=1,2$ and $U \sim U(0,1)$.
\item[Step 2.] Set $R=\min \{ T_1 , F_{T_{12}}^{-1} (U) \}$ and
$S=\min \{ T_2 , F_{T_{12}}^{-1} (1-U) \}$, where $F_{T_{12}}^{-1} (.)$ is the quantile function of  $T_{12} \sim E(\theta_{12})$.
\item[Step 3.] The desired pair is $(R,S)$.
\end{itemize} 
\end{algorithm} 
\begin{figure}[h!]
\includegraphics[width=\textwidth]{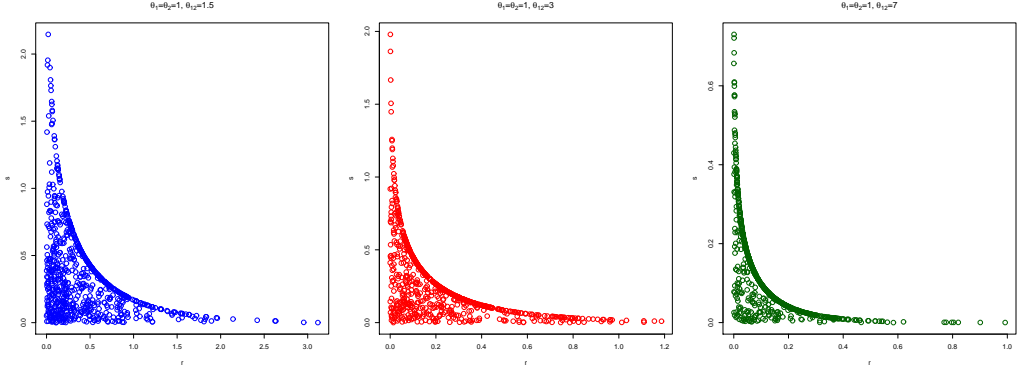}
\caption{Scatterplot of 750 generated data using Algorithm \ref{alg1} for different values of dependence parameter $\theta_{12}=1.5,3,7$ (from left to right) and fixed marginal parameters $\theta_1=\theta_2=1$.}
\label{scatter1}
\end{figure}
Figure \ref{scatter1} shows scatterplots of 750 generated data from Algorithm \ref{alg1}. As the dependence parameter $\theta_{12}$ increases, the dependence increases and so the data tend to assemble near (0,0).

%
%

\subsection{Estimation method}
Here, we will estimate the  parameters using maximum likelihood (ML) method.\\
Consider the random sample of size $m$, namely $\{(r_1,s_1), \ldots,(r_m,s_m)\}$ distributed from $BNMO(\theta_1,\theta_2,\theta_{12})$.
Let $m_1$ and $m_2$ denote the number of observations for which $e^{-\theta_{12} r}+e^{-\theta_{12} s}>1$ and $e^{-\theta_{12} r}+e^{-\theta_{12} s}=1$, respectively, such that $m_1 + m_2 = m$. The log-likelihood function for a given sample of observations is given as 
\begin{eqnarray}
l(\theta_1,\theta_2,\theta_{12}) &=&  - \theta_1 \sum_{j=1}^{m_1} r_j - \theta_2 \sum_{j=1}^{m_1} s_j  \nonumber \\
&&  +\sum_{j=1}^{m_1} \log \Big( \theta_2 (\theta_1 + \theta_{12})e^{- \theta_{12} r_j}  +\theta_1 (\theta_2 + \theta_{12})e^{- \theta_{12} s_j} - \theta_1 \theta_2 \Big) \nonumber \\
&&  + m_2 \log( \frac{\theta_{12}}{\theta_{1} + \theta_{12}} )   +\frac{\theta_2}{\theta_{12}} \sum_{j=m_1 + 1}^{m} \log (1-e^{-\theta_{12} r_j}),  \label{loglik} 
\end{eqnarray}
where the observations are classified such that $\{(r_1,s_1), \ldots,(r_{m_1},s_{m_1})\} \in A$ and $\{(r_{m_1 +1},s_{m_1 +1}), \ldots,(r_{m},s_{m})\} \in A^c$ and $A=\{ (r_i,s_i) | e^{-\theta_{12} r}+e^{-\theta_{12} s}>1 \}$.\\
Based on the normal equations (given in the Appendix), if either of $m_1$ or $m_2$ are zero, then the ML estimator may not be unique. However, this won't be an issue since
\begin{eqnarray}
P(m_1=0)&=&[P( e^{-\theta_{12} R}+e^{-\theta_{12} S}>1 )]^m \rightarrow 0 ~~ as~~m \rightarrow \infty, \nonumber 
\end{eqnarray}
and
\begin{eqnarray}
P(m_2=0)&=&[P( e^{-\theta_{12} R}+e^{-\theta_{12} S}=1 )]^m \rightarrow 0 ~~ as~~m \rightarrow \infty. \nonumber 
\end{eqnarray}
So, for moderate sample size m, the events  $[m_1=0]$ and $[m_2=0]$ are rare. For the case $m_1 , m_2 >0$, the resulting system of equations (normal equations given in the Appendix) cannot be solved in closed form expressions and so numerical methods are required. But, we found these methods to have less efficiency than the direct maximization of log-likelihood function in (\ref{loglik}). The maximization can be performed using \url{optim} function in the R software. Initial values for optimization are derived based on global non-linear optimization package "Rsolnp" in the R software version 3.6.1. The constraints $\theta_1,\theta_2,\theta_{12}>0$ was taken into account. We found the local maximums  after having different values of $\theta_1,\theta_2,\theta_{12}$. So, we select the global maximum based on the following relation
\begin{equation}
(\hat{\theta}_1,\hat{\theta}_2,\hat{\theta}_{12})={\arg\max}_{\theta_{1},\theta_{2},\theta_{12} \in \Theta} l(\theta_{1},\theta_{2},\theta_{12}).
\end{equation}
\subsection{Performance analysis}
Next, a finite sample performance of the estimators for marginal parameters $(\theta_{1},\theta_{2})$ and dependence parameter $\theta_{12}$ is given. The performance is evaluated according to bias and
mean square error (MSE) of the ML estimators introduced in the previous section. A specific sample size $m$ has been taken from $BNMO(1,3,0.8)$ and MSEs have been calculated based on 10000
iterations. The results are shown in Figure \ref{msebias}. Clearly, the ML estimator performs very well for small sample sizes.  Evidently, after some fluctuations, the values of bias becomes more stable around zero as the sample size increases. We must note that for the MLE, global maximum was unique all the time and did not correspond to the boundary of parameter space. The computational time required to identify the global
maximum after trying out all combinations of the initial values did not exceed 7 hours.
\begin{figure}[h!]
	\centering
	\includegraphics[width=\textwidth]{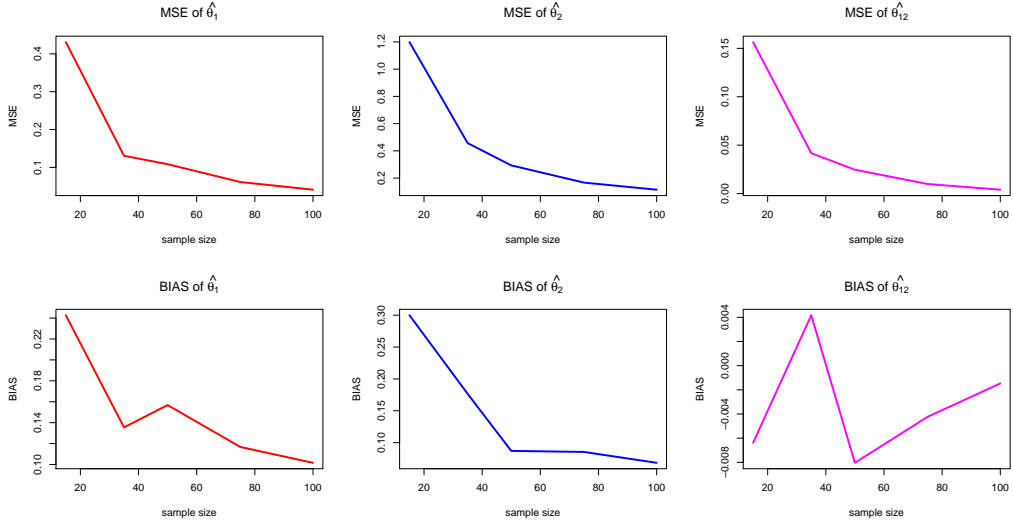}
	\caption{Performance analysis of ML estimators based on MSE and bias for $(\theta_1,\theta_2,\theta_{12}$)=(1,3,0.8) using 10000 independent replications.  }
	\label{msebias}
\end{figure}

\section{Application}
For illustrating results, an application of the BNGM distribution to a dataset is given in this section. \citet{mohsin2014new} explored the mercury (Hg) concentration in largemouth bass.  The data were collected from 53 different Florida lakes. They were used to examine the factors that influence the level of mercury concentration in bass. In specific dates, water samples were collected from the surface of the middle of each lake and the amount of alkalinity (mg/l), calcium (mg/l) and chlorophyll (mg/l) were measured in each sample. They used the average values of August and March. After that, a sample of fish was taken from each lake with sample sizes ranging from 4 to 44 fish and the minimum mercury concentration ($\mu$g/g) among the sampled fish were measured. \citet{lange1993influence} observed that the bio-accumulation of mercury in the largemouth bass was strongly influenced by the chemical characteristics of the lakes. Therefore, chemical substance like calcium along with minimum mercury concentration in the sampled fish is of interest. We use the proposed distribution to model these data. As \citet{mohsin2014new} stated, we have omitted the $40^{th}$ row of the data (considered as an outlier). A data summary is given in Table \ref{descriptive}. Based on the values of $\rho_s$ and $\tau$, both variables have moderate amount of dependency.  

%
\begin{table}
\centering
\caption{Descriptive statistics of data vectors Mercury and Calcium.}
\begin{tabular}{ccc}
\hline \hline
Statistics & Mercury & Calcium \\ 
\hline \hline
Minimum & 0.04 & 1.1 \\ 
$1^{st}$-Quantile & 0.09 & 3.3 \\ 
Median & 0.25 & 12.6 \\ 
Mean & 0.27 & 22.2 \\ 
$3^{rd}$-Quantile & 0.33 & 35.6 \\ 
Maximum & 0.92 & 90.7 \\ 
SD & 0.22 & 24.93 \\ 
Spearman's rho & \multicolumn{2}{c}{-0.536} \\ 
Kendall's tau & \multicolumn{2}{c}{-0.392} \\ 
$\frac{Rho}{Tau}$ & \multicolumn{2}{c}{1.36} \\ 
\hline 
\end{tabular} 
\label{descriptive}
\end{table}
\begin{figure}
	\centering
	\subfloat{{\includegraphics[width=0.5\textwidth]{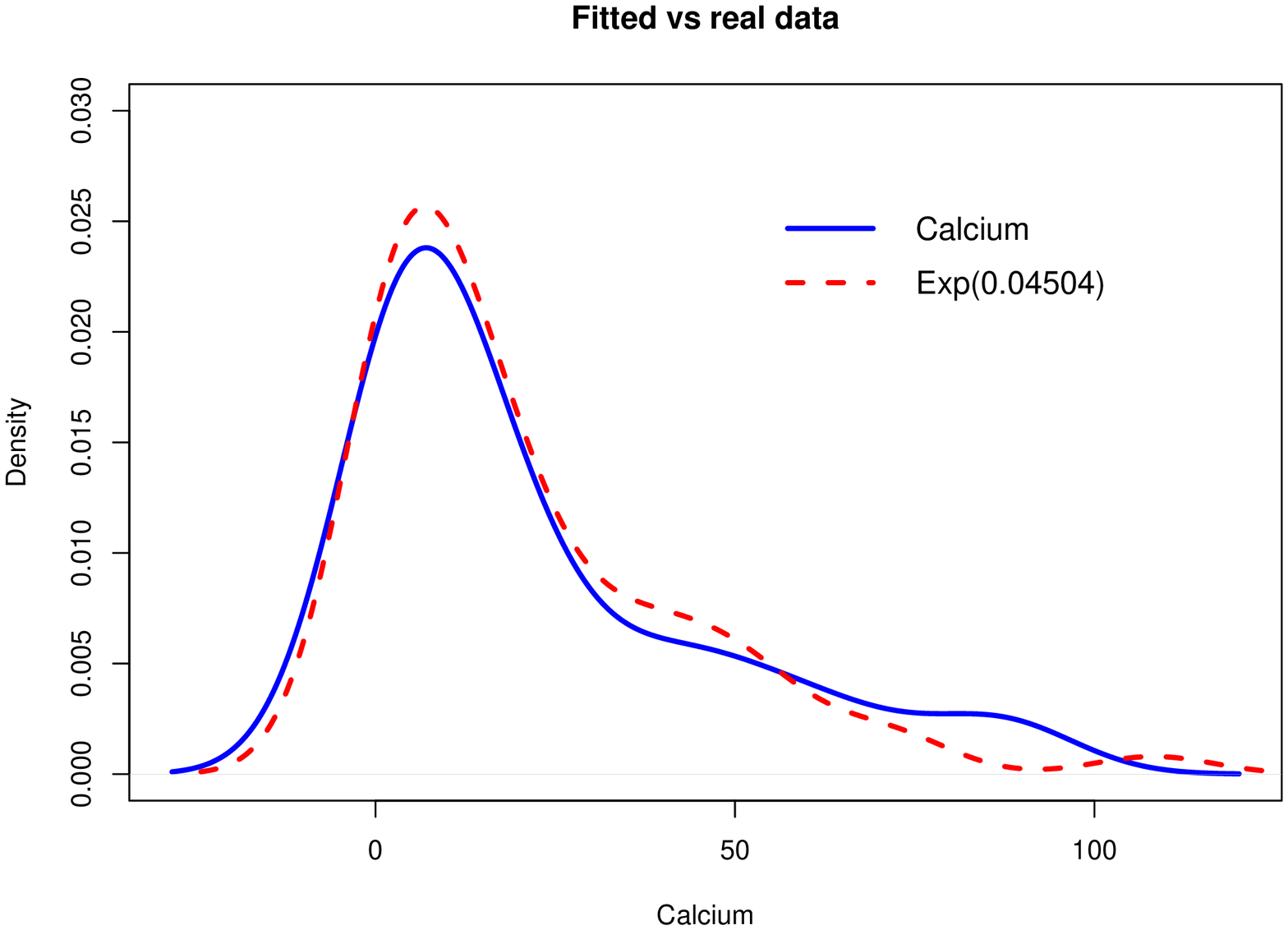} }}
	\subfloat{{\includegraphics[width=0.5\textwidth]{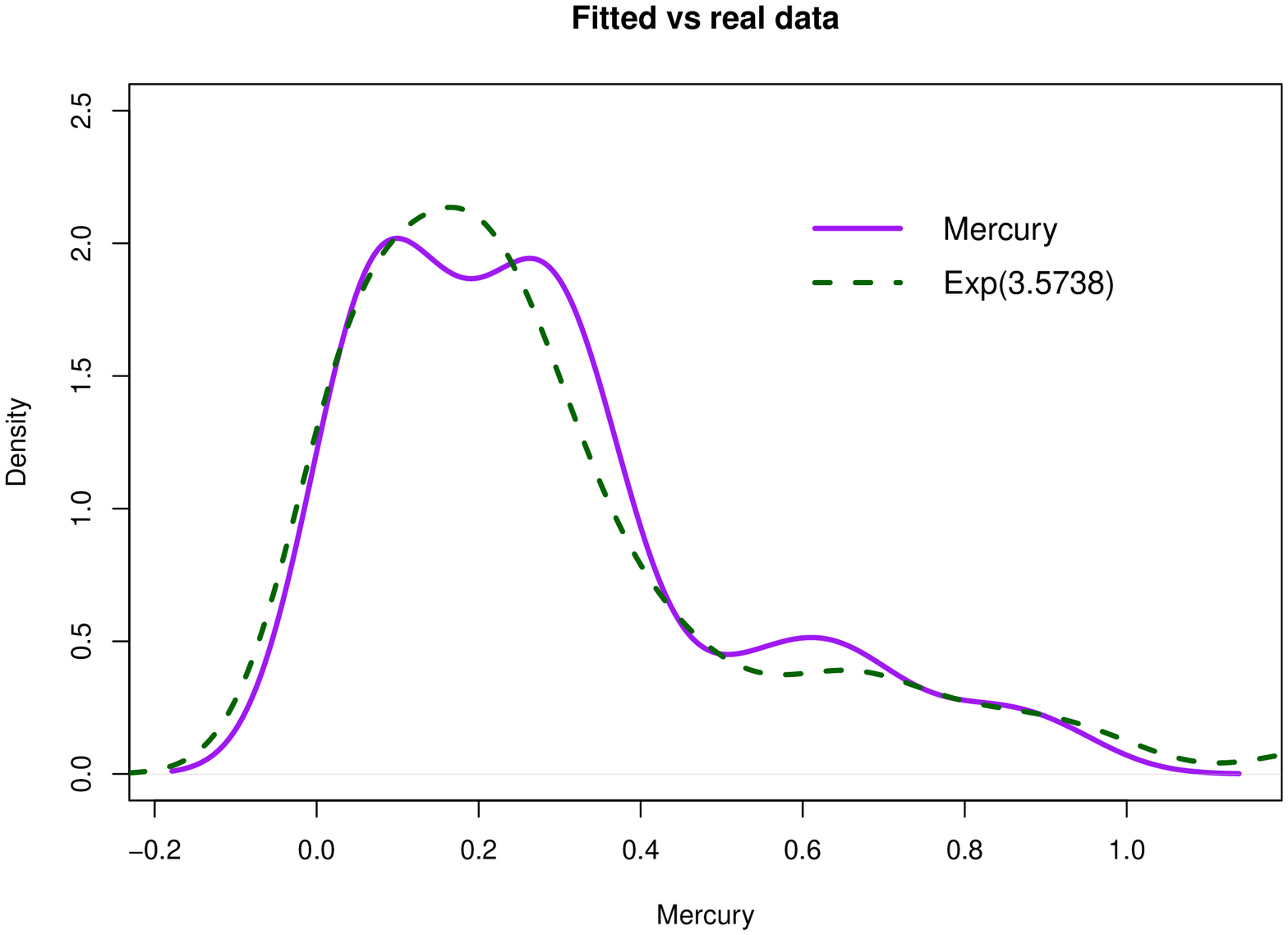} }}
	\caption{Density plot of Calcium (left) and Mercury (right) and their fitted distributions.}
	\label{density.fig}
\end{figure}
We have fitted an exponential distribution to the marginal data which are summarized in Table \ref{marginal.inference} and illustrated in Figure \ref{density.fig}. Clearly the marginal distributions are well fitted to the data.
\begin{table}
\centering
\caption{Marginal goodness-of-fit for Mercury and Calcium}
\begin{tabular}{ccccc}
\hline \hline
Variables & Distribution & MLE & Log-likelihood & K-S P-value  \\ 
\hline \hline 
Mercury & Exponential & 3.573 & 14.502 & 0.195 \\ 
Calcium & Exponential & 0.045 & -217.309 & 0.232 \\ 
\hline 
\end{tabular} 
\label{marginal.inference}
\end{table}
Now, that we are sure the marginal data are exponentially distributed, we are going to fit the joint model to the data (Mercury , Calcium) and compare it with the results given in \citet{mohsin2014new}. The results are given in Table \ref{bivariate.inference}. It is clear that the BNMO model is a better model than the BALE model given by \citet{mohsin2014new}. Both models are well fitted to the data based on the Kolmogrov-Smirnov goodness-of-fit criteria. Figure \ref{biv.scatter} shows the scatter plot of Mercury versus Calcium for real data and simulated data which are generated from the fitted BNMO model. 

\begin{table}
\centering
\caption{Goodness-of-fit for the joint vector (Mercury,Calcium).}
\begin{small}
\begin{tabular}{cccccc}
\hline \hline
Model & MLE & Log-Likelihood & K-S P-value. \\ 
\hline \hline
BNMO &  $\hat{\theta}_{1}=0.01, \hat{\theta}_2=3.67, \hat{\theta}_{12}=0.038$ & -194.0028 & 0.28 \\ 
BALE (\citet{mohsin2014new}) &  $\hat{\alpha}=3.63, \hat{\beta}=0.01, \hat{\gamma}=0.25$ & -3887.665 & 0.16 \\ 
\hline 
\end{tabular} 
\label{bivariate.inference}
\end{small}
\end{table}

\begin{figure}[h!]
	\centering
	\includegraphics[width=0.6\textwidth]{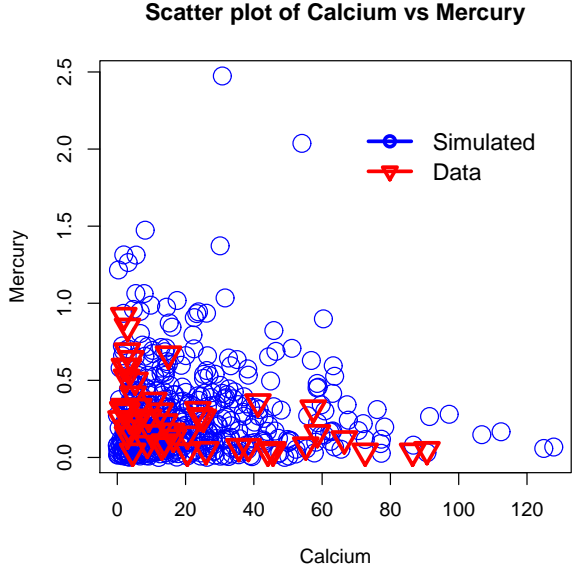}
	\caption{Scatterplot of Mercury versus Calcium for generated data which are from the best BNMO model. }
	\label{biv.scatter}
\end{figure}
\section{The multivariate case}
In this section we present the multivariate MO model covering all degrees of dependence. Consider the independent random variables $U \sim U(0,1)$, $T_{ij} \sim \bar{F}_{ij}$ distributed from the exponential distribution with mean $1/\theta_{ij}$. Let $\underline{\alpha}=(\alpha_{1,2},\ldots,\alpha_{n-1,n})$ be the vector of dependence structure vector of all joint elements where $\alpha_{ij}=+1$ indicates positive and   $\alpha_{ij}=-1$ states negative dependence structure for elements $i$ and $j$. Set $\underline{X}=(X_1,\ldots,X_n)$ as following:
\begin{equation}\label{mult.model}
\underline{X}=
\begin{cases}
X_1 = \min \{ T_{11} , T_{12}^* (\alpha_{1,2}) , \ldots , T_{1n}^* (\alpha_{1,n}) \} \\
X_2 = \min \{ T_{21}^* (\alpha_{1,2}), T_{22} , \ldots , T_{2n}^* (\alpha_{2,n}) \} \\
\ldots \\
X_n = \min \{ T_{n1}^* (\alpha_{1,n}),T_{n2}^* (\alpha_{2,n}) , \ldots , T_{nn} \} \\
\end{cases}.
\end{equation} 
When $\alpha_{ij}=+1$, set  $T_{ij}^* (\alpha_{i,j})=T_{ji}^* (\alpha_{i,j})=F_{ij}^{-1} (U)$ or $T_{ij}^* (\alpha_{i,j})=T_{ji}^* (\alpha_{i,j})=F_{ij}^{-1} (1-U)$. In the case  $\alpha_{ij}=-1$, put $T_{ij}^* (\alpha_{i,j})=F_{ij}^{-1} (U),T_{ji}^* (\alpha_{i,j})=F_{ij}^{-1} (1-U)$ or $T_{ij}^* (\alpha_{i,j})=F_{ij}^{-1} (1-U),T_{ji}^* (\alpha_{i,j})=F_{ij}^{-1} (U)$.\\
The vector $\underline{X}$ is distributed as the multivariate MO model covering all range of dependence. The survival function of $\underline{X}$  for observations $\underline{x}=(x_1,\ldots,x_n)$ is obtained as following:
\[ \bar{F}_{\underline{X}}(\underline{x}) = \prod_{j=1}^n P \big( T_{jj} > x_j \big) \prod_{i<j} P \big( T_{ij}^* (\alpha_{i,j}) > x_i , T_{ji}^* (\alpha_{i,j}) >x_j \big). \]
Clearly, for every $i<j$, we have
\begin{equation*}\label{mult.model}
P \big( T_{ij}^* (\alpha_{i,j}) > x_i , T_{ji}^* (\alpha_{i,j}) >x_j \big)=
\begin{cases}
\exp \{ - \theta_{ij} \max(x_i , x_j) \}; & \quad  \alpha_{ij}=+1 \\
e^{- \theta_{ij} x_i}  + e^{- \theta_{ij} x_j} -1; & \quad  \alpha_{ij}=-1
\end{cases}.
\end{equation*} 
with $e^{- \theta_{ij} x_i}  + e^{- \theta_{ij} x_j} \geq1$ for every $i<j$. So, 
\begin{equation*}
 \bar{F}_{\underline{X}}(\underline{x}) = \exp \Big\{ - \sum_{j=1}^n \theta_{jj} x_j \Big\}  \prod_{i<j} P \big( T_{ij}^* (\alpha_{i,j}) > x_i , T_{ji}^* (\alpha_{i,j}) >x_j \big). 
\end{equation*}
Note that, when $\alpha_{ij}=+1, \forall i<j$, we get the well-known multivariate MO model given in \citet{marshall1967multivariate}. For the case where every $\alpha_{ij}=-1, \forall i<j$, we get the new multivariate MO model with negative dependence structure (denoted by MNMO) as given:
 \begin{equation}\label{mult.model}
 \bar{F}_{\underline{X}}(\underline{x}) = \exp \Big\{ - \sum_{j=1}^n \theta_{jj} x_j \Big\}  \prod_{i<j} \Big( e^{- \theta_{ij} x_i}  + e^{- \theta_{ij} x_j} -1 \Big),
 \end{equation}
 where $e^{- \theta_{ij} x_i}  + e^{- \theta_{ij} x_j} \geq 1,\forall i<j $.
 \begin{example}
 	Consider the trivariate case ($n=3$) with $\alpha_{ij}=-1,\forall i<j$. Then, $\underline{X}$ can be obtained by
 	\begin{equation*}\label{examplemult}
 	\underline{X}=
 	\begin{cases}
 	X_1 = \min \{ T_{11} , F_{12}^{-1}(U), F_{13}^{-1}(U) \} \\
 	X_2 = \min \{ T_{22} , F_{12}^{-1}(1-U), F_{23}^{-1}(U) \} \\
 	X_3 = \min \{ T_{33} , F_{13}^{-1}(1-U), F_{23}^{-1}(1-U) \} \\
 	\end{cases}.
 	\end{equation*} 
 	The random vector given in (\ref{examplemult}) can be considered as a system with three components $(X_1,X_2,X_3)$ which are subject to joint shocks ($X_i$ and $X_j$ for $i<j$) where the shocks are likely to be distributed unequally within the pairs. The survival function of $\underline{X}$ in (\ref{examplemult}) is  
 	\begin{eqnarray}
 	\bar{F}_{\underline{X}}(x_1,x_2,x_3)&=& e^{-\theta_{11} x_1  -\theta_{22} x_1  -\theta_{33} x_3} \big( e^{-\theta_{12} x_1 } + e^{-\theta_{12} x_2} -1 \big) \nonumber \\
 	&& \times \big( e^{-\theta_{13} x_1 } + e^{-\theta_{13} x_3} -1 \big) \big( e^{-\theta_{23} x_2 } + e^{-\theta_{23} x_3} -1 \big) ,\nonumber   
 	\end{eqnarray}
 	where $e^{- \theta_{ij} x_i}  + e^{- \theta_{ij} x_j} \geq 1,\forall i<j $.
 \end{example}
 
The copula function gives us the raw dependence structure of a random vector that is independent from the marginal distributions. Based on the well-known Sklar theorem (see \citet{nelsen2007introduction}) for every random vector and copula function $C$, we have 
\[F(x_1 , \ldots , x_n) = C\Big(F_1(x_1) , \ldots , F_n(x_n)\Big),\] 
and consequently  for the survival copula $\hat{C}$ we get 
\[\bar{F}(x_1 , \ldots , x_n) = \hat{C}\Big(\bar{F}_1(x_1) , \ldots , \bar{F}_n(x_n)\Big).\] 
The survival copula associated with the model in (\ref{mult.model}) is 
\begin{equation}\label{www}
\hat{C}(u_1 , \ldots , u_n)=\prod_{j=1}^n u_j^{\gamma_{jj}} \prod_{i<j}  (u_i^{\gamma_{ij,i}} + u_j^{\gamma_{ij,j}} -1),
\end{equation} 
where $u_i^{\gamma_{ij,i}} + u_j^{\gamma_{ij,j}} \geq 1$, $\gamma_{jj}=\frac{\theta_{jj}}{ \theta_{j1}+\ldots+\theta_{jn} }$, $\gamma_{ij,i}=\frac{\theta_{ij}}{ \theta_{i1}+\ldots + \theta_{in} }$ and $\gamma_{ij,j}=\frac{\theta_{ij}}{ \theta_{j1}+\ldots + \theta_{jn} }$. For the case $n=2$, the copula in (\ref{www}) gives a special case of the model in \citet{khoudraji1995contribution} and \citet{dolati2014some} found some properties. 
\par Based on \citet{ghosh1981multivariate}, the random vector $(X_1,\ldots,X_n)$ is said to be right tail decreasing in sequence (RTDS), if for all real values $x_i$, $i=1,\ldots,n$ 
\begin{equation*}
P(X_i>x_i | X_1>x_1,\ldots,X_{i-1}>x_{i-1}),
\end{equation*}
is decreasing in $x_1,\ldots,x_{i-1}$. The concept RTDS establishes the negative dependence structure. On noting that a bivariate function $f$ is RR2, if for every $x_1<x_2$ and $y_1<y_2$, it holds that
\[ f(x_1,y_1) f(x_2,y_2) - f(x_1,y_2) f(x_2,y_1) \leq 0, \]
which is equivalent to $\frac{\partial^2 \ln(f(x,y))}{\partial x \partial y} \geq 0$. The following statement shows that the proposed model in (\ref{mult.model}) has negative dependence structure.
\begin{proposition}
Let $\underline{X}$ be distributed from the MNMO model in (\ref{mult.model}). Then, $\underline{X}$ is RTDS.
\end{proposition}
\begin{proof}
	Regarding to \citet{ghosh1981multivariate}, we know that the vector $(X_1,\ldots,X_n)$ is RTDS if the corresponding multivariate survival function $\bar{F}$ is RR2 in each pair of elements for fixed values of the remaining arguments. Since for every pair $(X_i,X_j)$, we have 
	\[ \frac{\partial^2 \ln(\bar{F}(x_i,x_j))}{\partial x_i \partial x_j} = \frac{-\theta_{12}^2e^{-\theta_{12}x_i-\theta_{12}x_j }}{(1-e^{-\theta_{12}x_i}-e^{-\theta_{12}x_j})^2} \leq 0. \]
	So, we conclude that  $\bar{F}$ is RR2 in each pair of elements $(X_i,X_j)$ and hence $(X_1,\ldots,X_n)$ is RTDS.
\end{proof}
\section{Conclusion}
In real applications, a system of components are often  exposed to different shocks. The amount of shocks are effective on the reliability of the system. Based on the well-known MO bivariate shock model in (\ref{MO.construction}), it is impossible to allocate the probability of the common shock ($T_{12}$) on each of components ($X_1$ and $X_2$). We have solved this issue by proposing a new MO shock model given in (\ref{eeer}) for bivariate and (\ref{mult.model}) for multivariate cases. The MO model in (\ref{MO.construction}) is a special case of the given model. Also, the obtained model has desirable properties such as covering positive and negative dependence structure and having closed form of stress-strength index making it useful in applications. There are not many bivariate exponential distributions with negative dependence structure and so this model is quite appealing with this regard. Having a singular component makes the new model challenging for estimating its parameters. We have given an estimation method and applied a performance analysis on the proposed estimator to see its effectiveness. The new model is used on the real data given \citet{mohsin2014new} (which is also a bivariate exponential distribution with negative structure) and we showed that our model is more promising than their model. Finally, we have proposed the multivariate case of the given model, followed by some of its properties.
\section{Appendix}
Let $\tilde{\theta}=(\theta_1,\theta_2,\theta_{12})$ and for all $j$:
$$\Delta_j =\theta_2 (\theta_1 + \theta_{12})e^{- \theta_{12} r_j}  +\theta_1 (\theta_2 + \theta_{12})e^{- \theta_{12} s_j} - \theta_1 \theta_2. $$ 
The Normal equations for estimating parameters are as following:
\begin{eqnarray}
\frac{\partial l(\tilde{\theta})}{\partial \theta_1}&=&-\sum_{j=1}^{m_1} r_j + \sum_{j=1}^{m_1} \frac{1}{\Delta_j} \Big( \theta_2 e^{- \theta_{12} r_j} +(\theta_2 + \theta_{12})e^{- \theta_{12} s_j} - \theta_2 \Big) - \frac{m_2}{ \theta_{1} + \theta_{12} }, \nonumber \\
\frac{\partial l(\tilde{\theta})}{\partial \theta_2}&=&-\sum_{j=1}^{m_1} s_j + \sum_{j=1}^{m_1} \frac{1}{\Delta_j} \Big( (\theta_1 + \theta_{12}) e^{- \theta_{12} r_j} +\theta_1e^{- \theta_{12} s_j} + \theta_1\Big) \nonumber \\
&& + \frac{1}{\theta_{12}} \sum_{j=m_1 +1}^m \log \Big( 1-\exp \{ - \theta_{12} r_j \} \Big), \nonumber \\
\text{and} \nonumber \\
 \nonumber \\
\frac{\partial l(\tilde{\theta})}{ \partial \theta_{12} }&=& 
\sum_{j=1}^{m_1}  \frac{-1}{\Delta_j} \Big(  \theta_2 ( \theta_1+ \theta_{12})  e^{- \theta_{12} r_j} r_j + \theta_1 ( \theta_2+ \theta_{12})  e^{- \theta_{12} s_j} s_j   \Big)  \nonumber \\
&& +\frac{m_2}{\theta_{12}} -\frac{m_2}{\theta_1+\theta_{12}} - \frac{\theta_2}{\theta_{12}^2}  \sum_{j=m_1 +1}^m \log \Big( 1-\exp \{ - \theta_{12} r_j \} \Big)  \nonumber \\
&& + \frac{\theta_2}{\theta_{12}} \sum_{j=m_1 +1}^m \frac{r_j e^{- \theta_{12} r_j}}{1-\exp \{ - \theta_{12} r_j \}}. \nonumber
\end{eqnarray}







\section*{References}
\bibliographystyle{apa}
\bibliography{refs}

\begin{thebibliography}{}

\bibitem[\protect\astroncite{Al-Mutairi et~al.}{2018}]{al2018weighted}
Al-Mutairi, D., Ghitany, M., and Kundu, D. (2018).
\newblock Weighted weibull distribution: Bivariate and multivariate cases.
\newblock {\em Brazilian Journal of Probability and Statistics}, 32(1):20--43.

\bibitem[\protect\astroncite{Balakrishnan}{2018}]{balakrishnan2018exponential}
Balakrishnan, K. (2018).
\newblock {\em Exponential distribution: Theory, methods and applications}.
\newblock Routledge.

\bibitem[\protect\astroncite{Basu and Sun}{1997}]{basu1997multivariate}
Basu, A.~P. and Sun, K. (1997).
\newblock Multivariate exponential distributions with constant failure rates.
\newblock {\em Journal of multivariate analysis}, 61(2):159--170.

\bibitem[\protect\astroncite{Bayramoglu and
  Ozkut}{2014}]{bayramoglu2014reliability}
Bayramoglu, I. and Ozkut, M. (2014).
\newblock The reliability of coherent systems subjected to marshall--olkin type
  shocks.
\newblock {\em IEEE Transactions on Reliability}, 64(1):435--443.

\bibitem[\protect\astroncite{Cha and Bad{\'\i}a}{2017}]{cha2017multivariate}
Cha, J.~H. and Bad{\'\i}a, F. (2017).
\newblock Multivariate reliability modelling based on dependent dynamic shock
  models.
\newblock {\em Applied Mathematical Modelling}, 51:199--216.

\bibitem[\protect\astroncite{Cherubini et~al.}{2015}]{cherubini2015marshall}
Cherubini, U., Durante, F., and Mulinacci, S. (2015).
\newblock Marshall-olkin distributions-advances in theory and applications.
\newblock {\em Springer Proceedings in Mathematics \& Statistics, Springer
  International Publishing}.

\bibitem[\protect\astroncite{Cherubini and
  Mulinacci}{2017}]{cherubini2017gumbel}
Cherubini, U. and Mulinacci, S. (2017).
\newblock The gumbel-marshall-olkin distribution.
\newblock In {\em Copulas and Dependence Models with Applications}, pages
  21--31. Springer.

\bibitem[\protect\astroncite{Cui and Li}{2007}]{cui2007analytical}
Cui, L. and Li, H. (2007).
\newblock Analytical method for reliability and mttf assessment of coherent
  systems with dependent components.
\newblock {\em Reliability Engineering \& System Safety}, 92(3):300--307.

\bibitem[\protect\astroncite{Dolati et~al.}{2014}]{dolati2014some}
Dolati, A., Mohseni, S., and {\'U}beda-Flores, M. (2014).
\newblock Some results on a transformation of copulas and quasi-copulas.
\newblock {\em Information Sciences}, 257:176--182.

\bibitem[\protect\astroncite{Elouerkhaoui}{2017}]{elouerkhaoui2017credit}
Elouerkhaoui, Y. (2017).
\newblock {\em Credit correlation: Theory and practice}.
\newblock Springer.

\bibitem[\protect\astroncite{Esary and Marshall}{1974}]{esary1974multivariate}
Esary, J.~D. and Marshall, A.~W. (1974).
\newblock Multivariate distributions with exponential minimums.
\newblock {\em The Annals of Statistics}, pages 84--98.

\bibitem[\protect\astroncite{Fan et~al.}{2009}]{fan2009multivariate}
Fan, J., Nunn, M.~E., and Su, X. (2009).
\newblock Multivariate exponential survival trees and their application to
  tooth prognosis.
\newblock {\em Computational Statistics \& Data Analysis}, 53(4):1110--1121.

\bibitem[\protect\astroncite{Genest et~al.}{2018}]{genest2018new}
Genest, C., Mesfioui, M., and Schulz, J. (2018).
\newblock A new bivariate poisson common shock model covering all possible
  degrees of dependence.
\newblock {\em Statistics \& Probability Letters}, 140:202--209.

\bibitem[\protect\astroncite{Ghosh and Ebrahimi}{1981}]{ghosh1981multivariate}
Ghosh, M. and Ebrahimi, N. (1981).
\newblock Multivariate negative dependence.
\newblock {\em Communications in Statistics-Theory and Methods},
  10(4):307--337.

\bibitem[\protect\astroncite{G{\'o}mez et~al.}{1998}]{gomez1998multivariate}
G{\'o}mez, E., Gomez-Viilegas, M., and Marin, J. (1998).
\newblock A multivariate generalization of the power exponential family of
  distributions.
\newblock {\em Communications in Statistics-Theory and Methods},
  27(3):589--600.

\bibitem[\protect\astroncite{Joe}{1997}]{joebook1997}
Joe, H. (1997).
\newblock {\em Multivariate models and multivariate dependence concepts}.
\newblock CRC Press.

\bibitem[\protect\astroncite{Khoudraji}{1995}]{khoudraji1995contribution}
Khoudraji, A. (1995).
\newblock {\em Contribution l'{\'e}tude des copules et la mod'elisation de
  valeurs extremes multivari{\'e}es}.
\newblock PhD thesis, PhD Thesis, Universit{\'e} de Laval, Qu{\'e}bec.

\bibitem[\protect\astroncite{Kundu et~al.}{2014}]{kundu2014multivariate}
Kundu, D., Franco, M., and Vivo, J.-M. (2014).
\newblock Multivariate distributions with proportional reversed hazard
  marginals.
\newblock {\em Computational Statistics \& Data Analysis}, 77:98--112.

\bibitem[\protect\astroncite{Kundu and Gupta}{2013}]{kundu2013bayes}
Kundu, D. and Gupta, A.~K. (2013).
\newblock Bayes estimation for the marshall--olkin bivariate weibull
  distribution.
\newblock {\em Computational Statistics \& Data Analysis}, 57(1):271--281.

\bibitem[\protect\astroncite{Kundu and Gupta}{2009}]{kundu2009bivariate}
Kundu, D. and Gupta, R.~D. (2009).
\newblock Bivariate generalized exponential distribution.
\newblock {\em Journal of Multivariate Analysis}, 100(4):581--593.

\bibitem[\protect\astroncite{Lange et~al.}{1993}]{lange1993influence}
Lange, T.~R., Royals, H., and Connor, L.~L. (1993).
\newblock Influence of water chemistry on mercury concentration in largemouth
  bass from florida lakes.
\newblock {\em Transactions of the American Fisheries Society}, 122(1):74--84.

\bibitem[\protect\astroncite{Li and Pellerey}{2011}]{li2011generalized}
Li, X. and Pellerey, F. (2011).
\newblock Generalized marshall--olkin distributions and related bivariate aging
  properties.
\newblock {\em Journal of Multivariate Analysis}, 102(10):1399--1409.

\bibitem[\protect\astroncite{Lin et~al.}{1993}]{lin1993multivariant}
Lin, H.-H., Chen, K., and Wang, R.-T. (1993).
\newblock A multivariant exponential shared-load model.
\newblock {\em IEEE Transactions on Reliability}, 42(1):165--171.

\bibitem[\protect\astroncite{Lindskog and McNeil}{2003}]{lindskog2003common}
Lindskog, F. and McNeil, A.~J. (2003).
\newblock Common poisson shock models: Applications to insurance and credit
  risk modelling.
\newblock {\em ASTIN Bulletin: The Journal of the IAA}, 33(2):209--238.

\bibitem[\protect\astroncite{Marshall and
  Olkin}{1967}]{marshall1967multivariate}
Marshall, A.~W. and Olkin, I. (1967).
\newblock A multivariate exponential distribution.
\newblock {\em Journal of the American Statistical Association},
  62(317):30--44.

\bibitem[\protect\astroncite{Mohsin et~al.}{2014}]{mohsin2014new}
Mohsin, M., Kazianka, H., Pilz, J., and Gebhardt, A. (2014).
\newblock A new bivariate exponential distribution for modeling moderately
  negative dependence.
\newblock {\em Statistical Methods \& Applications}, 23(1):123--148.

\bibitem[\protect\astroncite{Mohtashami-Borzadaran
  et~al.}{2020}]{mohtashami-borzadaran_jabbari_amini_2020}
Mohtashami-Borzadaran, H., Jabbari, H., and Amini, M. (2020).
\newblock Bivariate marshall--olkin exponential shock model.
\newblock {\em Probability in the Engineering and Informational Sciences},
  pages 1--21.

\bibitem[\protect\astroncite{Nelsen}{2007}]{nelsen2007introduction}
Nelsen, R.~B. (2007).
\newblock {\em An introduction to copulas}.
\newblock Springer Science \& Business Media.

\bibitem[\protect\astroncite{Raftery}{1984}]{raftery1984continuous}
Raftery, A.~E. (1984).
\newblock A continuous multivariate exponential distribution.
\newblock {\em Communications in Statistics-Theory and methods},
  13(8):947--965.

\bibitem[\protect\astroncite{Shih and Emura}{2016}]{shih2016bivariate}
Shih, J.-H. and Emura, T. (2016).
\newblock Bivariate dependence measures and bivariate competing risks models
  under the generalized fgm copula.
\newblock {\em Statistical Papers}, pages 1--18.

\bibitem[\protect\astroncite{Tawn}{1990}]{tawn1990modelling}
Tawn, J.~A. (1990).
\newblock Modelling multivariate extreme value distributions.
\newblock {\em Biometrika}, 77(2):245--253.

\end{thebibliography}
\end{document}